\documentclass[11pt]{llncs}

\usepackage{xspace}
\usepackage{array}
\usepackage{subfig}
\usepackage{xcolor}
\usepackage{epsfig}
\usepackage{gensymb}
\usepackage[english]{babel}
\usepackage{graphicx}
\usepackage{amssymb}
\usepackage{multirow}
\usepackage{todonotes}
\usepackage{enumerate}
\usepackage{paralist}
\usepackage{amsmath}
\usepackage{cases}
\usepackage{siunitx}
\usepackage{fullpage}

\newcommand{\remove}[1]{}

\renewcommand{\int}{int}

\renewenvironment{proof}
{{\bf Proof:}}{\hspace*{\fill}$\Box$\par\vspace{2mm}}

\newcommand{\areaSP}{O(n^{0.695})}

\begin{document}
\title{On the Edge-Length Ratio of Planar Graphs}
\author{Manuel Borrazzo and Fabrizio Frati}
\institute{University Roma Tre, Italy -- \email{\{borrazzo,frati\}@dia.uniroma3.it}\\
}
\maketitle

\begin{abstract}
	The edge-length ratio of a straight-line drawing of a graph is the ratio between the lengths of the longest and of the shortest edge in the drawing. The planar edge-length ratio of a planar graph is the minimum edge-length ratio of any planar straight-line drawing of the graph.
	
	In this paper, we study the planar edge-length ratio of planar graphs. We prove that there exist $n$-vertex planar graphs whose planar edge-length ratio is in $\Omega(n)$; this bound is tight. We also prove upper bounds on the planar edge-length ratio of several families of planar graphs, including series-parallel graphs and bipartite planar graphs.
\end{abstract}


\section{Introduction}\label{se:introduction}
The reference book for the graph drawing research field ``Graph Drawing: Algorithms for the Visualization of Graphs'', by Di Battista, Eades, Tamassia, and Tollis~\cite{dett-gd-93}, mentions that the minimization of the maximum edge length, provided that the minimum edge length is a fixed value, is among the most important aesthetic criteria that one should aim to satisfy in order to guarantee the readability of a graph drawing. A measure that naturally captures this concept is the \emph{edge-length ratio} of a drawing; this is defined as the ratio between the lengths of the longest and shortest \mbox{edge in the drawing.}

In this paper we are interested in the construction of planar straight-line drawings with small edge-length ratio. From an algorithmic point of view, it has long been known that deciding whether a graph admits a planar straight-line drawing with edge-length ratio equal to $1$ is an NP-hard problem. This was first proved by Eades and Wormald~\cite{ew-fel-90} for biconnected planar graphs and then by Cabello et al.~\cite{DBLP:journals/jgaa/CabelloDR07} for triconnected planar graphs. From a combinatorial point of view, the study of planar straight-line drawings with small edge-length ratio started only recently, when Lazard, Lenhart, and Liotta~\cite{DBLP:journals/tcs/LazardLL19} proved that every outerplanar graph admits a planar straight-line drawing with edge-length ratio smaller than $2$ and that, for every fixed $\epsilon>0$, there exist outerplanar graphs whose every planar straight-line drawing has edge-length ratio larger than $2-\epsilon$.

Adopting the notation and the definitions from~\cite{DBLP:conf/gd/LazardLL17,DBLP:journals/tcs/LazardLL19}, we denote by $\rho(\Gamma)$ the edge-length ratio of a straight-line drawing $\Gamma$ of a graph $G$, i.e., $\rho(\Gamma)=\max\limits_{e_1,e_2\in E(G)}\frac{\ell_{\Gamma}(e_1)}{\ell_{\Gamma}(e_2)}$, where $\ell_{\Gamma}(e)$ denotes the Euclidean length of the segment representing an edge $e$ in $\Gamma$. The \emph{planar edge-length ratio} $\rho(G)$ of $G$ is the minimum edge-length ratio of any planar straight-line drawing of $G$. We prove \mbox{the following results.}

First, we prove that there exist $n$-vertex planar graphs whose planar edge-length ratio is in $\Omega(n)$. This bound is asymptotically tight, as every planar graph admits a planar straight-line drawing on an $O(n)\times O(n)$ grid~\cite{DBLP:journals/combinatorica/FraysseixPP90,DBLP:conf/soda/Schnyder90}; such a drawing has edge-length ratio in $O(n)$. While our lower bound is not surprising, it was unexpectedly challenging to prove it. Some classes of graphs which are often used in order to prove lower bounds for graph drawing problems turn out to have constant planar edge-length ratio; see Figure~\ref{fig:examples}.

\begin{figure}[tb]\tabcolsep=4pt
	\centering
	\begin{tabular}{c c}
		\includegraphics[scale=1]{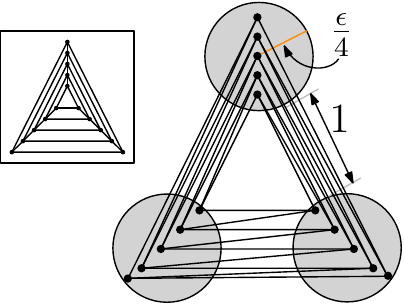} \hspace{3mm} &
		\includegraphics[scale=1.1]{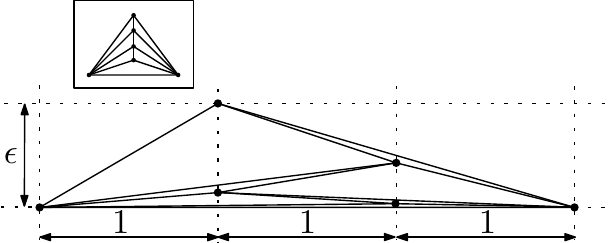} \\
		(a) \hspace{3mm} & (b)\\
	\end{tabular}
	\caption{(a) A drawing with edge-length ratio smaller than $1+\epsilon$ of the nested-triangle graph. (b) A drawing with edge-length ratio smaller than $3$ of the plane $3$-tree obtained as the join of a path with an edge.}
	\label{fig:examples}
\end{figure}

Second, we provide upper bounds for the planar edge-length ratio of several families of planar graphs. Namely, we prove that plane $3$-trees have planar edge-length ratio bounded by their ``depth'' and that, for every fixed $\epsilon>0$, bipartite planar graphs have planar edge-length ratio smaller than $1+\epsilon$. Most interestingly, we prove that every $n$-vertex graph with treewidth at most two, including $2$-trees and series-parallel graphs, has sub-linear planar edge-length ratio; our upper bound is $O(n^{\log_2 \phi})\subseteq \areaSP$, where $\phi=\frac{1+\sqrt 5}{2}$ is the golden ratio. Lazard et al.~\cite{DBLP:journals/tcs/LazardLL19} asked whether the planar edge-length ratio of $2$-trees is bounded by a constant. Very recently, Blazej et al.~\cite{bfll-elr2t-19,20-bfll-elr2t} answered the above question in the negative; namely, they proved that there exist $n$-vertex $2$-trees whose planar edge-length ratio is in $\Omega(\log n)$. Thus, our upper bound provides a significant counterpart to Blazej et al.'s result; further, our result sharply contrasts with the fact that there exist $n$-vertex $2$-trees for which every planar straight-line drawing on a grid requires an edge to have length in $\Omega(n)$~\cite{DBLP:journals/dmtcs/Frati10}.

The paper is organized as follows. In Section~\ref{se:preliminaries}, we introduce some definitions; in Section~\ref{se:planar}, we prove a lower bound for the planar edge-length ratio of planar graphs; in Section~\ref{se:classes}, we prove upper bounds for the planar edge-length ratio of families of planar graphs; finally, in Section~\ref{se:conclusions}, we conclude and present some open problems.

\section{Definitions and preliminaries}\label{se:preliminaries}

In this section we establish some definitions and preliminaries.

A {\em drawing} of a graph represents each vertex as a point in the plane and each edge as an open curve between its end-vertices. A drawing is \emph{straight-line} if each edge is represented by a straight-line segment. A drawing is {\em planar} if no two edges intersect, except at common end-vertices. A planar drawing of a graph defines connected regions of the plane, called {\em faces}. The only unbounded face is the {\em outer face}, while the other faces are {\em internal}. Two planar drawings of a (connected) graph are \emph{equivalent} if: (i) the clockwise order of the edges incident to each vertex is the same in both drawings; and (ii) the clockwise order of the edges along the boundary of the outer face is the same in both drawings. A {\em plane embedding} is an equivalence class of planar drawings and a {\em plane graph} is a graph with a prescribed plane embedding. Throughout the paper, whenever we talk about a planar drawing of a plane graph $G$, we always assume, even when not explicitly stated, that it respects the plane embedding associated with~$G$; further, we sometimes talk about a face of a plane graph $G$, meaning a face of a planar drawing respecting the plane embedding associated with~$G$.

For any two distinct points $a$ and $b$ in the plane, we denote by $\overline{ab}$ the straight-line segment between $a$ and $b$ and by $||\overline{ab}||$ the Euclidean length of such a segment. For any three distinct and noncollinear points $a$, $b$, and $c$ in the plane, we denote by $abc$ the triangle whose vertices are $a$, $b$, and $c$. Further, for a triangle $\Delta$, we denote by $p(\Delta)$ its perimeter and by $\angle_a(\Delta)$ the angle at a vertex $a$ of $\Delta$.

We will use the following lemma more than once.

\begin{lemma} \label{le:subgraphs}
	Let $G$ be a planar graph and $G'$ be a subgraph of $G$. \mbox{Then $\rho(G')\leq \rho(G)$.}
\end{lemma}

\begin{proof}
	Consider any planar straight-line drawing $\Gamma$ of $G$ with $\rho(\Gamma)=\rho(G)$. Let $\Gamma'$ be the planar straight-line drawing of $G'$ obtained from $\Gamma$ by removing the vertices and edges that are not in $G'$. The length of the shortest edge in $\Gamma'$ is larger than or equal to the length of the shortest edge in $\Gamma$ (the inequality is strict if no edge of $G$ whose length in $\Gamma$ is minimum belongs to $G'$). Analogously, the length of the longest edge in $\Gamma'$ is smaller than or equal to the length of the longest edge in $\Gamma$. It follows that $\rho(\Gamma')\leq \rho(\Gamma)$, which implies the statement.
\end{proof}

\section{A Lower Bound for Planar Graphs}\label{se:planar}

In this section we show that there exist $n$-vertex planar graphs whose planar edge-length ratio is in $\Omega(n)$. This lower bound is the strongest possible. Namely, every $n$-vertex planar graph admits a planar straight-line drawing on a $O(n)\times O(n)$ grid~\cite{DBLP:journals/combinatorica/FraysseixPP90,DBLP:conf/soda/Schnyder90}; such a drawing has edge-length ratio in~$O(n)$. 

\begin{theorem} \label{th:lower-bound}
	For every $n=6k-2$ with $k\in \mathbb N_{>0}$, there exists an $n$-vertex planar graph whose planar edge-length ratio is in $\Omega(n)$.
\end{theorem}

The rest of the section is devoted to the proof of Theorem~\ref{th:lower-bound}. The graphs that we use in order to prove the lower bound of Theorem~\ref{th:lower-bound} form a subclass of the planar $3$-trees, for which we will show an upper bound on the planar edge-length ratio in Section~\ref{se:plane-3-trees}.


\begin{figure}[htb]\tabcolsep=4pt
	\centering
	\includegraphics[scale=0.75]{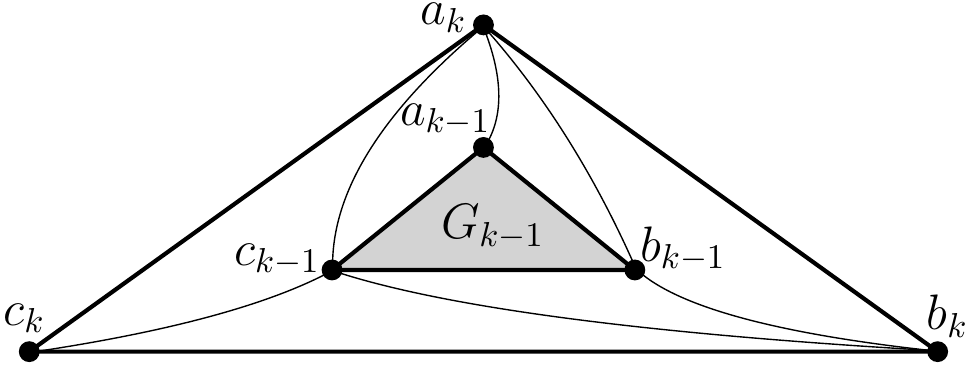} 
	\caption{Construction of the graph $G_{k}$ from the graph $G_{k-1}$.}
	\label{fig:lower-graph}
\end{figure}

We start by defining a class of $3k$-vertex plane graphs, for which we will show that any planar straight-line drawing in which the outer face is delimited by a prescribed cycle has edge-length ratio in $\Omega(k)$. For a $3$-cycle $\mathcal C$ in a plane graph $G$, we denote by $abc$ the clockwise order in which the vertices $a$, $b$, and $c$ of $\mathcal C$ occur along $\mathcal C$. For any integer $k\geq 1$, we define a $3k$-vertex plane graph $G_k$ as follows; refer to Fig.~\ref{fig:lower-graph}. Let $G_1$ coincide with a $3$-cycle ${\mathcal C_1}=a_1b_1c_1$. Now suppose that, for some integer $k\geq 2$, a plane graph $G_{k-1}$ has been defined so that its outer face is delimited by a $3$-cycle ${\mathcal C_{k-1}}=a_{k-1}b_{k-1}c_{k-1}$. Let $G_k$ consist of
\begin{inparaenum}[(i)]
	\item a $3$-cycle ${\mathcal C_{k}}=a_{k}b_{k}c_{k}$,
	\item the plane graph $G_{k-1}$, embedded inside ${\mathcal C_{k}}$, and
	\item the edges $a_ka_{k-1}$, $a_kb_{k-1}$, $a_kc_{k-1}$, $b_kb_{k-1}$, $b_kc_{k-1}$, $c_kc_{k-1}$.
\end{inparaenum} 
Note that $G_k$ has $3k$ vertices.

We prove a lower bound for the edge-length ratio of any planar straight-line drawing~$\Gamma$ of $G_k$ in which the outer face is delimited by ${\mathcal C_{k}}$. Assume, w.l.o.g.\ up to a uniform scaling of $\Gamma$, that the length of the shortest edge is $1$. We prove that, for $i=1,2,\dots,k$, the perimeter $p(\Delta_i)$ of the triangle $\Delta_i$ representing ${\mathcal C_i}$ in $\Gamma$ is at least $\gamma\cdot i$, for a constant $\gamma$ to be determined later. This implies that $p(\Delta_k)\in \Omega(k)$, hence the longest of the three segments composing $\Delta_k$ has length in $\Omega(k)$, and the edge-length ratio $\rho(\Gamma)$ of $\Gamma$ is in $\Omega(k)$. 

The perimeter $p(\Delta_1)$ of $\Delta_1$ is at least $3$, given that each of the three segments composing $\Delta_1$ has length greater than or equal to $1$. Now assume that $p(\Delta_{i-1})\geq \gamma\cdot (i-1)$, for some integer $i\geq 2$ and some constant $\gamma\leq 3$. We prove that $p(\Delta_{i})\geq p(\Delta_{i-1}) + \gamma$, which implies that $p(\Delta_{i})\geq \gamma\cdot i$. 

Before proceeding, we need two geometric lemmata. Refer to Fig.~\ref{fig:lower-graph-lemmata}. Let $a$, $b$, and $c$ be the three vertices of a triangle $\Delta$ and let $d$ be a point outside $\Delta$ such that $a$ either lies inside the triangle $\Delta'$ with vertices $b$, $c$, and $d$, or it lies in the interior of $\overline{bd}$, or it lies in the interior of $\overline{cd}$.

\begin{figure}[htb]\tabcolsep=4pt
	\centering
	\includegraphics[scale=1]{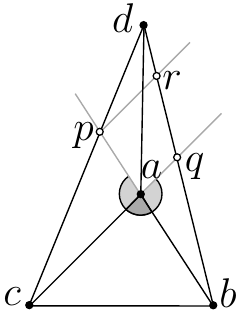}
	\caption{Illustration for the proofs of Lemmata~\ref{le:triangle-inside-triangle} and~\ref{le:lowerbound-small-angle}.}
	\label{fig:lower-graph-lemmata}
\end{figure}

\begin{lemma} \label{le:triangle-inside-triangle}
	$p(\Delta')> p(\Delta)$. 
\end{lemma}

\begin{proof}
	If $a$ lies in the interior of $\overline{cd}$ (of $\overline{bd}$), then the triangular inequality implies that $||\overline{bd}||+||\overline{da}||>||\overline{ba}||$ (resp.\ $||\overline{cd}||+||\overline{da}||>||\overline{ca}||$),\mbox{ hence $p(\Delta')>p(\Delta)$. }
	
	If $a$ lies inside $\Delta'$, then let $p$ be the intersection point of the straight line through $a$ and $b$ with $\overline{cd}$. By the triangular inequality, we have $||\overline{ap}||+||\overline{cp}||>||\overline{ac}||$, hence $p(bcp)>p(\Delta)$. Again by the triangular inequality, we have $||\overline{bd}||+||\overline{dp}||>||\overline{bp}||$, hence $p(\Delta')>p(bcp)$. 
\end{proof}

We remark that a stronger version of Lemma~\ref{le:triangle-inside-triangle}, which we do not need here, is in fact true: For any two convex polygons $P$ and $Q$ such that $Q$ is contained inside $P$, the perimeter of $P$ is larger than the perimeter of $Q$.

\begin{lemma} \label{le:lowerbound-small-angle}
	If $||\overline{ad}||\geq 1$ and $\angle_a (\Delta)\leq 90\degree$, then $p(\Delta')> p(\Delta)+1$. 
\end{lemma}

\begin{proof}
	Suppose first that $a$ lies in the interior of $\overline{cd}$. Since $\angle_a (\Delta)\leq 90\degree$, we have that $\angle_a (bad)\geq 90\degree$, hence $||\overline{bd}||>||\overline{ba}||$. It follows that $p(\Delta')- p(\Delta)= ||\overline{bd}||+||\overline{ad}||-||\overline{ba}||>1$. The case in which $a$ lies in the interior of $\overline{bd}$	is analogous.
	
	Suppose next that $a$ lies inside $\Delta'$. Let $p$ be the intersection point of the straight line through $a$ and $b$ with $\overline{cd}$, and	let $q$ be the intersection point of the straight line through $a$ and $c$ with $\overline{bd}$. Since $\angle_a (\Delta)\leq 90\degree$, we have that $\angle_a (cap)=\angle_a (baq)\geq 90\degree$, hence $||\overline{cp}||>||\overline{ca}||$ and $||\overline{bq}||>||\overline{ba}||$. It follows that $p(\Delta')-p(\Delta)>||\overline{dp}||+||\overline{dq}||$. 
	
	We claim that $||\overline{dp}||>||\overline{aq}||$. Let $r$ be the intersection point between $\overline{bd}$ and the line passing through $p$ that is parallel to the line through $a$ and $c$. The triangles $baq$ and $bpr$ are similar, hence $\angle_{p} (bpr)=\angle_{a} (baq)\geq  90\degree$. Thus, $\angle_{r} (bpr)<90\degree$ and $\angle_{r} (dpr)>90\degree$. It follows that $||\overline{dp}||> ||\overline{pr}||$; further, $||\overline{pr}||>||\overline{aq}||$, again by the similarity of the triangles $baq$ and $bpr$. This proves the claim. It can be analogously proved that $||\overline{dq}||>||\overline{ap}||$.
	
	By the triangular inequality, we have $||\overline{ap}||+||\overline{dp}||>||\overline{ad}||$ and $||\overline{aq}||+||\overline{dq}||>||\overline{ad}||$, hence $||\overline{ap}||+||\overline{dp}||+||\overline{aq}||+||\overline{dq}||>2||\overline{ad}||\geq 2$. Since $||\overline{dp}||>||\overline{aq}||$ and $||\overline{dq}||>||\overline{ap}||$, it follows that $||\overline{dp}||+||\overline{dq}||>1$.
\end{proof}

We now return the proof that $p(\Delta_{i})\geq p(\Delta_{i-1}) + \gamma$. Refer to Fig.~\ref{fig:largeangle-1}. Assume, w.l.o.g.\ up to a rotation of the Cartesian axes, that $\overline{b_{i-1}c_{i-1}}$ is horizontal, with $b_{i-1}$ to the right of $c_{i-1}$ and with $a_{i-1}$ above them. Let $\Delta'_{i-1}$ and $\Delta''_{i-1}$ be the triangles $a_ib_{i-1}c_{i-1}$ and $a_i b_i c_{i-1}$ in $\Gamma$, respectively. By Lemma~\ref{le:triangle-inside-triangle}, we have  $p(\Delta_{i})> p(\Delta''_{i-1})>p(\Delta'_{i-1})>p(\Delta_{i-1})$. 
\begin{figure}[htb]
	\centering
	\includegraphics[scale=1.2]{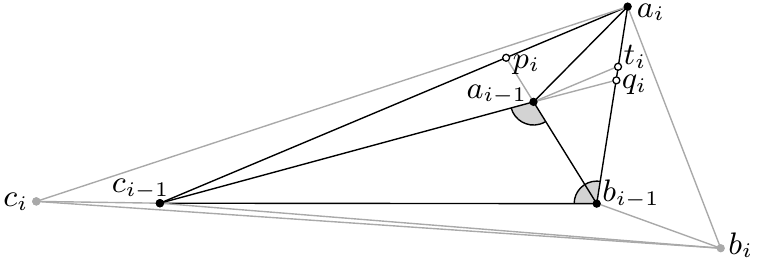}	
	\caption{Illustration for the proof that $p(\Delta_{i})\geq p(\Delta_{i-1}) + \gamma$.}
	\label{fig:largeangle-1}
\end{figure}
If $\angle_{a_{i-1}}(\Delta_{i-1})\leq 90\degree$, then by Lemma~\ref{le:lowerbound-small-angle} we have $p(\Delta'_{i-1})>p(\Delta_{i-1})+1$ and thus $p(\Delta_{i})>p(\Delta_{i-1})+1$ and we are done, as long as $\gamma\leq 1$. We can hence assume that $\angle_{a_{i-1}}(\Delta_{i-1})> 90\degree$; this implies that $a_{i-1}$ is to the left of the vertical line $\ell_{b}$ through $b_{i-1}$. Further, if $\angle_{b_{i-1}}(\Delta'_{i-1})\leq 90\degree$, then by Lemma~\ref{le:lowerbound-small-angle} we have $p(\Delta''_{i-1})>p(\Delta'_{i-1})+1$ and thus $p(\Delta_{i})>p(\Delta_{i-1})+1$ and we are done, as long as $\gamma\leq 1$. We can hence assume that $\angle_{b_{i-1}}(\Delta'_{i-1})> 90\degree$; this implies that $a_{i}$ is to the right of $\ell_{b}$. 

Let $p_i$ be the intersection point of the straight line through $a_{i-1}$ and $b_{i-1}$ with $\overline{c_{i-1}a_i}$, and let $q_i$ be the intersection point of the straight line through $a_{i-1}$ and $c_{i-1}$ with $\overline{b_{i-1}a_{i}}$. 

Assume first that $||\overline{a_iq_i}||\geq 0.4$. By Lemma~\ref{le:triangle-inside-triangle}, we have that $p(b_{i-1}c_{i-1}q_i)>p(\Delta_{i-1})$; further, since $\angle_{q_i}(c_{i-1}q_ia_i)>\angle_{b_{i-1}}(c_{i-1}b_{i-1}a_i)>90\degree$, we have that $||\overline{c_{i-1}a_i}||>||\overline{c_{i-1}q_i}||$, hence $p(\Delta'_{i-1})=p(b_{i-1}c_{i-1}q_i) + ||\overline{c_{i-1}a_i}|| + ||\overline{a_iq_i}|| - ||\overline{c_{i-1}q_i}||>p(\Delta_{i-1})+0.4$ and we are done, as long as $\gamma\leq 0.4$. 

Assume next that $||\overline{a_iq_i}||<0.4$. We show that this implies that $||\overline{a_ip_i}||\geq 0.4$. Suppose, for a contradiction, that $||\overline{a_ip_i}||< 0.4$. Consider the intersection point $t_i$ of $\overline{b_{i-1}a_i}$ with the line through $a_{i-1}$ parallel to the line through $c_{i-1}$ and $a_i$. Since $\angle_{q_i}(a_{i-1}q_it_i)=\angle_{q_i}(c_{i-1}q_ia_i)>90\degree$, we have $||\overline{a_{i-1}q_i}||<||\overline{a_{i-1}t_i}||$. Further, by the similarity of the triangles $b_{i-1}a_ip_i$ and $b_{i-1}t_ia_{i-1}$, we have $||\overline{a_{i-1}t_i}||<||\overline{a_{i}p_i}||$. Hence, $||\overline{a_{i-1}q_i}||<0.4$. Then the triangular inequality implies that $||\overline{a_{i-1}a_i}||< ||\overline{a_{i-1}q_i}||+||\overline{a_{i}q_i}||<0.8$, while $||\overline{a_{i-1}a_i}||\geq 1$, given that $a_{i-1}a_i$ is an edge of $G_i$, a contradiction. We can hence assume that $||\overline{a_ip_i}||\geq 0.4$. 

For the sake of simplicity of notation, let $x= ||\overline{b_{i-1}a_i}||$, $y=||\overline{a_ip_i}||$, and $z=||\overline{b_{i-1}p_i}||$. By Lemma~\ref{le:triangle-inside-triangle}, we have $p(b_{i-1}c_{i-1}p_i)>p(\Delta_{i-1})$, hence $p(\Delta'_{i-1})-p(\Delta_{i-1})>x+y-z$.  Note that $x\geq 1$, since $b_{i-1}a_i$ is an edge of $G_i$, and $y\geq 0.4$, by assumption. Recall that $\angle_{b_{i-1}}(\Delta'_{i-1})> 90\degree$, which implies that $\angle_{a_i}(\Delta'_{i-1})=\angle_{a_i}(b_{i-1}a_ip_i) < 90\degree$. By the Pythagorean Inequality, we have that $z<\sqrt{x^2+y^2}$, hence $p(\Delta'_{i-1})-p(\Delta_{i-1})>x+y-\sqrt{x^2+y^2}$. The derivative $\frac{\partial(x+y-\sqrt{x^2+y^2})}{\partial x}=\frac{\sqrt{x^2+y^2}-x}{\sqrt{x^2+y^2}}$ is positive for every real value of $x$ and $y$; the same is true for the derivative $\frac{\partial(x+y-\sqrt{x^2+y^2})}{\partial y}$. Hence, the minimum value of $x+y-\sqrt{x^2+y^2}$ is achieved when $x$ and $y$ are minimum, that is, when $x=1$ and $y=0.4$. With such values we get $x+y-\sqrt{x^2+y^2}> 0.32$. Hence, $p(\Delta'_{i-1})>p(\Delta_{i-1})+0.32$ and we are done, as long as $\gamma\leq 0.32$.  

By picking $\gamma=0.3$, we conclude the proof that $p(\Delta_{i})\geq p(\Delta_{i-1}) + \gamma$, which implies that $p(\Delta_{k})\in \Omega(k)$ and hence that $\rho(\Gamma)\in \Omega(k)$. 

{\bf Allowing for a different outer face.} Finally, we show how the above lower bound on the edge-length ratio of planar straight-line drawings in which the outer face is delimited by a prescribed cycle can be used in order to prove Theorem~\ref{th:lower-bound}. Consider the complete graph $K_4$ on four vertices, say $a$, $b$, $c$, and $d$. Further, consider two copies $G'_k$ and $G''_k$ of the $3k$-vertex plane graph $G_k$; let $\mathcal C'_k$ and $\mathcal C''_k$ denote the copies of the cycle $\mathcal C_k$ in $G'_k$ and $G''_k$, respectively. Glue $G'_k$ and $G''_k$ with $K_4$ by identifying the $3$-cycle $abc$ with $\mathcal C'_k$ and the $3$-cycle $abd$ with $\mathcal C''_k$. Denote by $G$ the resulting $n$-vertex planar graph and note that $G$ has $6k-2$ vertices. In any planar drawing $\Gamma$ of $G$, the planar drawing of $G'_k$ has its outer face delimited by $\mathcal C'_k$ or the planar drawing of $G''_k$ has its outer face delimited by $\mathcal C''_k$, hence $\rho(\Gamma)\in \Omega(k)$. The proof of Theorem~\ref{th:lower-bound} is concluded by observing that $k\in \Omega(n)$.   

\section{Upper Bounds for Planar Graph Classes}\label{se:classes}

In this section we prove upper bounds for the planar edge-length ratio of various families of planar graphs. 

\subsection{Plane $\bf 3$-Trees} \label{se:plane-3-trees}

A {\em plane $3$-tree} is a maximal plane graph that can be constructed as follows. The only plane $3$-tree with $3$ vertices is a $3$-cycle embedded in the plane. For $n\geq 4$, an $n$-vertex plane $3$-tree $G$ is obtained from an $(n-1)$-vertex plane $3$-tree $G'$ by inserting a vertex $v$ inside an internal face $f$ of $G'$ and by connecting $v$ to the three vertices of $G'$ incident to $f$. See Figure~\ref{fig:3-trees}(a). A \emph{planar $3$-tree} is a planar graph that admits a plane embedding as a \mbox{plane $3$-tree.}

\begin{figure}[htb]\tabcolsep=4pt
	\centering
	\begin{tabular}{c c}
		\includegraphics[scale=1.25]{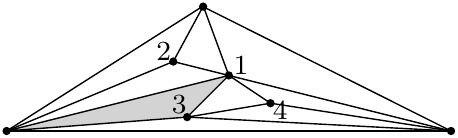} \hspace{1mm} &
		\includegraphics[scale=1.25]{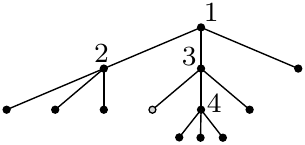}\\
		(a) \hspace{1mm} & (b)
	\end{tabular}
	\caption{(a) A plane $3$-tree $G$. (b) The tree $T_G$ associated with $G$. The leaf of $T_G$ representing the gray face of $G$ is also gray.}
	\label{fig:3-trees}
\end{figure}

An $n$-vertex plane $3$-tree $G$ is naturally associated with a rooted ternary tree $T_G$ whose internal nodes represent the internal vertices of $G$ and whose leaves represent the internal faces of $G$ ($T_G$ is called {\em representative tree} of $G$ in~\cite{DBLP:journals/jgaa/MondalNRA11}). Formally, $T_G$ is defined as follows; refer to Figure~\ref{fig:3-trees}(b). If $n=3$, then $T_G$ consists of a single node, representing the unique internal face of $G$. If $n>3$, then let $G'$ be a plane $3$-tree such that $G$ can be obtained by inserting a vertex $v$ inside an internal face $f$ of $G'$ and by connecting $v$ to the three vertices of $G'$ incident to $f$. Let $t_f$ be the leaf representing $f$ in $T_{G'}$. Then $T_{G}$ is obtained from $T_{G'}$ by inserting three new leaves as children of $t_f$. In $T_{G}$, the node $t_f$ represents $v$ and its children represent the faces of $G$ incident to $v$. The {\em depth} of $T_G$ is the maximum number of nodes in any root-to-leaf path in $T_{G}$. The \emph{depth} of $G$ is the depth of $T_{G}$. Note that the depth of an $n$-vertex plane $3$-tree is in $\Omega(\log n)$ and in $O(n)$. We have the following.

\begin{theorem} \label{th:planar-3-trees}
	Every plane $3$-tree with depth $k$ has planar edge-length \mbox{ratio in $O(k)$.}
\end{theorem}

\begin{proof}
	The \emph{$x$-extension} (the \emph{$y$-extension}) of a triangle $\Delta$ is the difference between the maximum and the minimum $x$-coordinate (resp.\ $y$-coordinate) of a vertex of $\Delta$. The $x$-extension and the $y$-extension of a segment are defined analogously. 
	
	Let $G$ be any plane $3$-tree with depth $k$. Fix any constant $\epsilon>0$ and represent the $3$-cycle $\cal C$ delimiting the outer face of $G$ as any triangle $\Delta$ whose $y$-extension is $\epsilon$ and whose three sides have $x$-extension equal to $1$, $k$, and $k+1$. 
	
	Now assume that we have constructed a drawing $\Gamma'$ of a plane $3$-tree $G'$ which is a subgraph of $G$ that includes $\cal C$. Assume that $\Gamma'$ satisfies the following invariant: every internal face $f$ of $G'$ is delimited by a triangle that has one side whose $x$-extension is equal to $1$, one side whose $x$-extension is greater than or equal to $k_f$, and one side whose $x$-extension is greater than or equal to $k_f+1$, where $k_f$ is the depth of the subtree of $T_G$ rooted at the node corresponding to $f$. Initially, this is the case with $G'=\mathcal C$ and $\Gamma'=\Delta$; note that the only internal face of $G'$ corresponds to the root of $T_G$, which has depth $k$.
	
	\begin{figure}[htb]\tabcolsep=4pt
		\centering
		\includegraphics[scale=1.25]{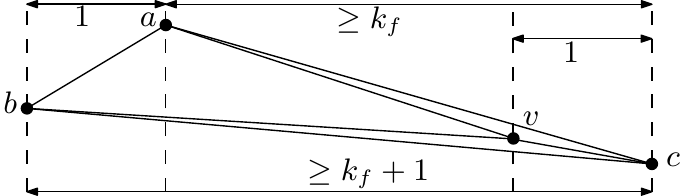}
		\caption{Inserting a vertex $v$ of $G$ in a face $f$ of $G'$.}
		\label{fig:3-trees-construction}
	\end{figure}
	
	Let $t_f$ be any leaf of $T_{G'}$ which is not a leaf of $T_G$. Let $f$ be the internal face of $G'$ represented by $t_f$ in $T_{G'}$, let $\Delta_f$ be the triangle representing $f$ in $\Gamma'$, let $abc$ be the $3$-cycle delimiting $f$, and let $v$ be the internal vertex of $G$ represented by $t_f$ in $T_{G}$; see Figure~\ref{fig:3-trees-construction}. By the invariant, we can assume that the $x$-extensions of $\overline{ab}$, $\overline{ac}$, and $\overline{bc}$ are equal to $1$, greater than or equal to $k_f$, and greater than or equal to $k_f+1$, respectively. Place $v$ inside $f$ in $\Gamma'$ so that the $x$-extension of $\overline{vc}$ is equal to $1$ and draw the edges $va$, $vb$, and $vc$ as straight-line segments. This results in a planar straight-line drawing $\Gamma''$ of a plane $3$-tree $G''$ which is a subgraph of $G$ and which has one more vertex than $G'$. Note that the invariant is satisfied by $\Gamma''$; in particular, $\overline{av}$ and $\overline{bv}$ have $x$-extension greater than or equal to $k_f-1$ and greater than or equal to $k_f$, respectively, hence each face $f'$ of $G''$ incident to $v$ is delimited by a triangle whose sides have  the desired $x$-extensions, given that the subtree of $T_G$ rooted at the node corresponding to $f'$ has depth $k_f-1$.
	
	Eventually, we get a planar straight-line drawing of $G$ such that every edge has length at least $1$, given that it has $x$-extension greater than or equal to $1$, and at most $k+1+\epsilon\in O(k)$, given that it has $x$-extension smaller than or equal to $k+1$ and $y$-extension smaller than or equal to $\epsilon$.
\end{proof}

The bound in Theorem~\ref{th:planar-3-trees} is tight. Namely, Theorem~\ref{th:lower-bound} shows that there exists a plane $3$-tree that has depth $k$ and for which any planar straight-line drawing has edge-length ratio in $\Omega(k)$ (even if the drawing is not required to respect the prescribed plane embedding). Further, Theorem~\ref{th:planar-3-trees} implies that any {\em balanced} $n$-vertex plane $3$-tree, i.e., a plane $3$-tree $G$ such that $T_G$ is a balanced tree, has planar edge-length ratio in $O(\log n)$.

\subsection{$\bf 2$-Trees} \label{se:planar-2-trees}

For any integer $n\geq 2$, an $n$-vertex {\em $2$-tree} $G$ is a graph whose vertex set has an ordering $v_1,v_2,\dots,v_n$ such that $v_1v_2$ is an edge of $G$, called {\em root} of $G$, and, for $i=3,\dots,n$, the vertex $v_i$ has exactly two neighbors $p(v_i)$ and $q(v_i)$ in $\{v_1,v_2,\dots,v_{i-1}\}$, where $p(v_i)$ and $q(v_i)$ are adjacent in $G$. The vertices $v_3,v_4,\dots,v_n$, i.e., the vertices of $G$ not in its root are called \emph{internal}. For an edge $v_iv_j$ of $G$, an \emph{apex} of $v_iv_j$ is a vertex $v_k$, with $k>i$ and $k>j$, such that $p(v_k)=v_i$ and $q(v_k)=v_j$; further, the \emph{side edges} of $v_iv_j$ are all the edges $v_iv_k$ and $v_jv_k$ such that $v_k$ is an apex of $v_iv_j$; finally, an edge $v_iv_j$ is \emph{trivial} if it has no apex, otherwise it is \emph{nontrivial}. Most of this section is devoted to a proof of the following theorem.


\begin{theorem} \label{th:2-trees}
	Every $n$-vertex $2$-tree has planar edge-length ratio in $O(n^{\log_2 \phi})\subseteq \areaSP$, where $\phi=\frac{1+\sqrt 5}{2}$ is the golden ratio.
\end{theorem}

In the following, we first define a family of $2$-trees, which we call \emph{linear $2$-trees}, and show that they admit drawings with constant edge-length ratio (see Section~\ref{subsubse:linear2trees}). We will later show how to find, in any $2$-tree $G$, a subgraph which is a linear $2$-tree and whose removal splits $G$ into ``small'' components. This decomposition, together with the drawing algorithm for linear $2$-trees, will be used in order to construct a planar straight-line drawing of $G$ with sub-linear edge-length ratio \mbox{(see Section~\ref{subsubse:general2trees}).}

\subsubsection{Linear $\bf 2$-Trees} \label{subsubse:linear2trees}

A {\em linear $2$-tree} is a $2$-tree such that every edge has at most one nontrivial side edge; see Figure~\ref{fig:linear-2-trees}(a). We now classify the vertices of a linear $2$-tree $H$ into vertices of class $1$, class $2$, and class $3$, so that every edge of $H$ has its end-vertices in different classes. First, $v_1$ and $v_2$ are vertices of class $1$ and class $2$, respectively, where $v_1v_2$ is the root of $H$. Now we repeatedly consider an edge $uv$ of $H$ such that $u$ and $v$ have been already classified and the apexes of $uv$ have not been classified yet. We let every apex be in the unique class different from the classes of $u$ and $v$. Based on the classification of the vertices of $H$, we also classify the edges of $H$ into edges of class $1$-$2$, class $1$-$3$, and class $2$-$3$, where an edge is of class $a$-$b$ if its end-vertices are of classes $a$ and $b$.

\begin{figure}[htb]\tabcolsep=4pt
	\centering
	\begin{tabular}{c c c}
		\includegraphics[scale=1.25]{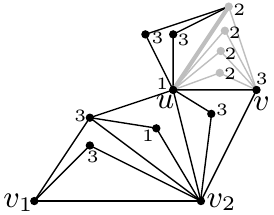} &
		\includegraphics[scale=1.25]{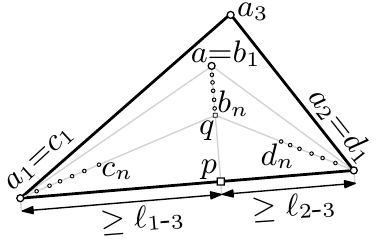} &
		\includegraphics[scale=1.25]{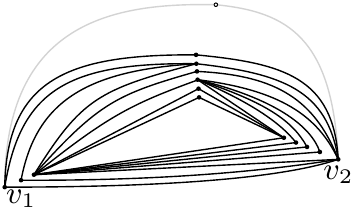} \\
		(a) & (b) & (c)
	\end{tabular}
	\caption{(a) A linear $2$-tree $H$. The apexes and the side edges of the edge $uv$ are gray; the only nontrivial side edge of $uv$ is represented by a thick line segment. The numbers show the classes of the vertices. (b) The points $b_1,b_2,\dots, b_n,q,c_1,c_2,\dots,c_n,d_1,d_2,\dots,d_n$ inside $a_1a_2a_3$ (for the sake of readability, there are fewer points than there should be). (c) The drawing of $H$ constructed by the algorithm L2T-drawer (for the sake of readability, some edges are represented \mbox{by nonstraight curves).}}
	\label{fig:linear-2-trees}
\end{figure}

We now give an algorithm, called {\bf L2T-drawer}, that constructs a planar straight-line drawing $\Gamma_H$ of a linear $2$-tree $H$. The algorithm L2T-drawer receives as input a triangle $a_1a_2a_3$ and three real values $\ell_{1\textrm{-}2},\ell_{1\textrm{-}3},\ell_{2\textrm{-}3} \geq 1$ such that $\ell_{1\textrm{-}3} + \ell_{2\textrm{-}3}\leq ||\overline{a_1a_2}||$ and such that $\ell_{1\textrm{-}2}<||\overline{a_1a_2}||$. The algorithm constructs a planar straight-line drawing $\Gamma_H$ of $H$ such that the following \mbox{properties are satisfied:}

\begin{enumerate}[(L1)]
	\item for $i=1,2$, the vertex $v_i$ lies at $a_i$;
	\item every internal vertex of $H$ lies inside $a_1a_2a_3$; and
	\item the length of every edge of class $x$-$y$ is at least $\ell_{x\textrm{-}y}$, for each $x\textrm{-}y\in \{1\textrm{-}2,1\textrm{-}3,2\textrm{-}3\}$.
\end{enumerate}


Refer to Figure~\ref{fig:linear-2-trees}(b). Let $a$ be a point inside $a_1a_2a_3$ such that the line through $a$ orthogonal to $\overline{a_1a_2}$ intersects $\overline{a_1a_2}$ in a point $p$ with $||\overline{a_1p}||\geq \ell_{1\textrm{-}3}$ and $||\overline{a_2p}||\geq \ell_{2\textrm{-}3}$; this exists because of the assumption $\ell_{1\textrm{-}3} + \ell_{2\textrm{-}3}\leq ||\overline{a_1a_2}||$. Let $\epsilon = \min \{||\overline{aa_1}||-\ell_{1\textrm{-}3},||\overline{aa_2}||-\ell_{2\textrm{-}3}, ||\overline{a_1a_2}||-\ell_{1\textrm{-}2}\}$ and note that $\epsilon>0$. Let $b_1,b_2,\dots,b_n,q$ be $n+1$ points on the straight-line segment $\overline{ap}$, in this order from $a$ to $p$, such that $||\overline{aq}||\leq \frac{\epsilon}{3}$. Further, let $c_1=a_1,c_2,\dots,c_n$ be $n$ points on the straight-line segment $\overline{a_1q}$, in this order from $a_1$ to $q$, such that $||\overline{a_1c_n}||\leq \frac{\epsilon}{3}$. Finally, let $d_1=a_2,d_2,\dots,d_n$ be $n$ points on the straight-line segment $\overline{a_2q}$, in this order from $a_2$ to $q$, such that $||\overline{a_2d_n}||\leq \frac{\epsilon}{3}$.

The algorithm L2T-drawer works as follows. Refer to Figure~\ref{fig:linear-2-trees}(c). We initialize $\Gamma_H$ by drawing the root $v_1v_2$ of $H$ as the straight-line segment $\overline{a_1a_2}$, where $a_i$ represents $v_i$, for $i=1,2$. Now L2T-drawer proceeds in steps. During one step, all the apexes and side edges of a single nontrivial edge of $H$ are drawn. The algorithm maintains the invariant that, before each step, $\Gamma_H$ is a planar straight-line drawing of an $m$-vertex subgraph $H_m$ of $H$ such that the following properties are satisfied for some integers $j,k,l$ with $m=j+k+l$: 

\begin{enumerate}[(i)]
	\item the vertices of $H_m$ of classes $1$, $2$, and $3$ are drawn at the points $c_1,\dots,c_k$, at the points $d_1,\dots,d_l$, and at the points $b_1,\dots,b_j$, respectively;
	\item if $H_m$ does not coincide with $H$, then there is exactly one edge $e_m$ that is a nontrivial edge of $H$, that is in $H_m$, and whose apexes are not in $H_m$; and
	\item if $H_m$ does not coincide with $H$, then the end-vertices of $e_m$ lie at $b_j$ and $c_k$, or at $c_k$ and $d_l$, or at $b_j$ and $d_l$. 
\end{enumerate}

The invariant is indeed satisfied after the initialization of $\Gamma_H$ to a drawing of $v_1v_2$, with $m=2$, $k=l=1$, and $j=0$. 

We now perform one step. By Property~(ii), if $H_m$ does not coincide with $H$, then there is exactly one edge $e_m$ that is a nontrivial edge of $H$, that is in $H_m$, and whose apexes are not in $H_m$. Assume that $e_m$ is a $1$-$2$ edge and hence, by Properties~(i) and~(iii), its end-vertices lie at $c_k$ and $d_l$; the other cases are analogous. If $j>0$, then Property~(i) implies that $e_m$ is on the boundary of the triangle $b_j c_k d_l$, while the rest of $\Gamma_H$ is in the closure of the exterior of $b_j c_k d_l$ (if $j=0$, then $\Gamma_H$ consists only of the drawing of the edge $e_m=v_1v_2$). Draw the $x\geq 1$ apexes of $e_m$, which are vertices of class $3$, at the points $b_{j+1},\dots,b_{j+x}$, so that the only nontrivial side edge $e_{m+x}$ of $e_m$, if any, is incident to the apex drawn at $b_{j+x}$. Draw the side edges of $e_m$ as straight-line segments. If $j=0$, then, by construction, the side edges of $e_m$ do not cross each other and do not cross $e_m$. If $j>0$, then the points $b_{j+1},\dots,b_{j+x}$ are internal to the triangle $b_j c_k d_l$, and thus so are the side edges of $e_m$, which hence do not cross the rest of $\Gamma_H$. It follows that, after this step, $\Gamma_H$ is a planar straight-line drawing of an $(m+x)$-vertex subgraph $H_{m+x}$ of $H$ satisfying the invariant. In particular, at most one side edge $e_{m+x}$ of $e_m$ is nontrivial in $H$, given that $H$ is a linear $2$-tree. 


Eventually, the algorithm constructs a planar straight-line drawing $\Gamma_H$ of $H$. By construction, the vertices $v_1$ and $v_2$ are placed at $a_1$ and $a_2$, respectively, hence $\Gamma_H$ satisfies property (L1). Further, the vertices of $H$ different from $v_1$ and $v_2$ are placed at the points $b_1,b_2,\dots,b_n,c_2,c_3,\dots,c_n,$ $d_2,d_3,\dots,d_n$, which are inside $a_1a_2a_3$, by construction, hence $\Gamma_H$ satisfies property (L2). Finally, we prove that $\Gamma_H$ satisfies property (L3). Consider any edge of class $1$-$3$, which is represented by a straight-line segment $\overline{c_kb_j}$ in $\Gamma_H$. By the triangular inequality we have $||\overline{c_kb_j}||\geq ||\overline{aa_1}||-||\overline{a_1c_k}||-||\overline{ab_j}||\geq \ell_{1\textrm{-}3}+\epsilon-\frac{2\epsilon}{3}>\ell_{1\textrm{-}3}$. It can be analogously proved that any edge of class $2$-$3$ has length larger than $\ell_{2\textrm{-}3}$ and that any edge of class $1$-$2$ has length larger than $\ell_{1\textrm{-}2}$ in $\Gamma_H$. 

%
%
%
%
%

\subsubsection{General $\bf 2$-Trees} \label{subsubse:general2trees}

We now deal with general $2$-trees. Let $G$ be a $2$-tree and let $v_1v_2$ be its root. Consider any subgraph $H$ of $G$ that is a linear $2$-tree and that has $v_1v_2$ as its root. For any edge $uv$ of $H$ we define an \emph{$H$-component} $G_{uv}$ of $G$ as follows. Remove from $G$ the vertices of $H$ and their incident edges; this splits $G$ into several connected components and we let $G_{uv}$ be the $2$-tree which is the subgraph of $G$ induced by $u$, by $v$, and by the vertex sets of the connected components containing a vertex adjacent to both $u$ and $v$. See Figure~\ref{fig:2-trees}(a). The edge $uv$ is the root of $G_{uv}$. An $H$-component of $G$ is of \emph{class}~$1$-$2$,~$1$-$3$, or~$2$-$3$ if its root is of class~$1$-$2$,~$1$-$3$, or~$2$-$3$, respectively. \mbox{See Figure~\ref{fig:2-trees}(b).}

\begin{figure}[htb]\tabcolsep=4pt
	\centering
	\begin{tabular}{c c}
		\includegraphics[scale=1.2]{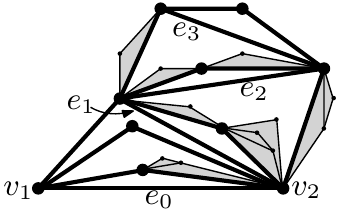}&
		\includegraphics[scale=1.2]{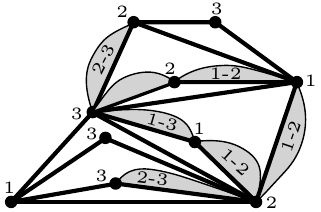}\\
		(a) & (b)
	\end{tabular}
	\caption{(a) A $2$-tree $G$ and a subgraph $H$ of $G$ which is a linear $2$-tree; the vertices and edges of $H$ are represented by larger disks and thicker lines, respectively. The $H$-components of $G$ are shown within shaded regions. (b) A schematic representation of $G$ and $H$; each vertex of $H$ and each $H$-component of $G$ is labeled with its class.}
	\label{fig:2-trees}
\end{figure}

For technical reasons, we let $n$ be the number of vertices of $G$ \emph{minus one}. The plan is: (1) to find a subgraph $H$ of $G$ that is a linear $2$-tree, that has $v_1v_2$ as its root, and such that every $H$-component of $G$ has ``few'' internal vertices; (2) to construct a planar straight-line drawing $\Gamma_H$ of $H$ by means of the algorithm \mbox{L2T-drawer}; and (3) to recursively draw each $H$-component independently, plugging such drawings into $\Gamma_H$, thus obtaining a drawing of $G$. We start with the following lemma, which draws inspiration from a technique for decomposing ordered binary trees proposed by Chan~\cite{DBLP:journals/algorithmica/Chan02}.

\begin{lemma}\label{le:h-component}
	There exists a subgraph $H$ of $G$ that is a linear $2$-tree, that has $v_1v_2$ as its root, and that satisfies the following property. Let $x$, $y$, and $z$ be the maximum number of vertices of an $H$-component of $G$ of class $1$-$3$, $2$-$3$, and $1$-$2$, respectively, minus one. Then $z\leq \frac{n}{2}$; further, we have that: (i) $x\leq \frac{n}{2}$ and $y\leq \frac{n-x}{2}$, or (ii) $y\leq \frac{n}{2}$ and $x\leq \frac{n-y}{2}$.
\end{lemma}

\begin{proof}
	We give an algorithm to find the required subgraph $H$ of $G$. We are going to define a sequence $H_0,H_1,\dots$ of subgraphs of $G$, where $H_i$ is a subgraph of $H_{i+1}$, for $i=0,1,\dots$; the desired graph $H$ is the last graph in this sequence. Together with the sequence $H_0,H_1,\dots$, we are also going to define a sequence of {\em designated edges} $e_0,e_1,\dots$, where $e_{i+1}$ is a side edge of $e_i$, for $i=0,1,\dots$. This is done so to maintain the following invariants. First, for $i=0,1,\dots$, the designated edge $e_i$, the apexes of $e_i$, and the side edges of $e_i$ all belong to $H_i$. Second, no apex of a side edge of $e_i$ belongs to $H_i$. The sequences $H_0,H_1,\dots$ and $e_0,e_1,\dots$ are initialized by defining $e_0=v_1v_2$ and by defining $H_0$ as the graph consisting of $e_0$, as well as of the apexes and the side edges of $e_0$. Note that the invariants are satisfied by the definition of $H_0$ and $e_0$.
	
	Now assume that, for some integer $i\geq 0$, sequences $H_0,H_1,\dots,H_i$ and $e_0,e_1,\dots,e_i$ have been defined, so that the invariants are satisfied. Two cases are possible. If the designated edge $e_i$ has no side edges, then $H$ coincides with $H_i$. Otherwise, a side edge of $e_i$ is chosen as the new designated edge $e_{i+1}$ and $H_{i+1}$ is obtained by adding all the apexes and side edges of $e_{i+1}$ to $H_i$. The rule to determine which side edge of $e_i$ is $e_{i+1}$ is the following. Let $K_{i+1}$ be the $H_i$-component of $G$ whose root is a side edge of $e_i$ and whose number of vertices is maximum (ties are broken arbitrarily); then $e_{i+1}$ is the root of $K_{i+1}$. 
	
	It remains to prove that $H$ satisfies the requirements of the lemma. Note that $H$ has $v_1v_2$ as its root, by construction. Further, the invariant that no apex of a side edge of $e_i$ belongs to $H_i$ ensures that $H$ is a linear $2$-tree. 
	
	In order to complete the proof we exploit the following properties.
	
	\begin{enumerate}[(P1)]
		\item No designated edge $e_i$ is the root of an $H$-component of $G$.
		\item The root of any $H$-component of $G$ is the side edge of a \mbox{designated edge $e_i$.}
	\end{enumerate}
	
	Both properties are trivially satisfied by $H_0$ and are easily shown to be satisfied by $H_{i+1}$ given that  they are satisfied by $H_i$.
	
	\begin{figure}[htb]\tabcolsep=4pt
		\centering
		\begin{tabular}{c c c c}
			\includegraphics[scale=1.15]{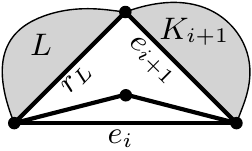} \hspace{1mm} &
			\includegraphics[scale=1.15]{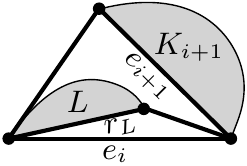} \hspace{1mm} &
			\includegraphics[scale=1.15]{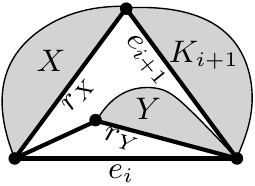} \hspace{1mm} &
			\includegraphics[scale=1.15]{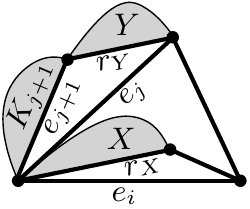}\\
			(a) \hspace{1mm} & (b)\hspace{1mm} & (c)\hspace{1mm} & (d)
		\end{tabular}
		\caption{Illustrations for the proof of Lemma~\ref{le:h-component}. (a)-(b) Every $H$-component of $G$ has at most $\frac{n}{2}+1$ vertices; in (a) $r_L$ and $e_{i+1}$ share a vertex, while in (b) they do not. (c)-(d) If $r_X$ and $r_Y$ are side edges for $e_i$ and $e_j$ with $i\leq j$, then $x\leq \frac{n}{2}$ and $y\leq \frac{n-x}{2}$; in (c) we have $i=j$, while in (d) we have $i<j$.}
		\label{fig:decomposition}
	\end{figure}
	
	We prove that every $H$-component of $G$ has at most $\frac{n}{2}+1$ vertices; refer to Figures~\ref{fig:decomposition}(a) and~\ref{fig:decomposition}(b). Suppose, for a contradiction, that there exists an $H$-component $L$ of $G$ that has more than $\frac{n}{2}+1$ vertices, hence more than $\frac{n}{2}-1$ internal vertices. By (P2) the root $r_L$ of $L$ is a side edge of a designated edge $e_i$. By (P1) we have $e_{i+1}\neq r_L$. By construction, $K_{i+1}$ has a number of (internal) vertices which is larger than or equal to the number of (internal) vertices of $L$. Since $L$ and $K_{i+1}$ do not share internal vertices, they have a total of more than $n-2$ internal vertices; these are also internal vertices of $G$. Further, by (P1) we have $e_0\neq r_L$, hence at least one end-vertex of $r_L$ is an internal vertex of $G$. It follows that $G$ has more than $n-1$ internal vertices, a contradiction. Hence, every $H$-component of $G$ has at most $\frac{n}{2}+1$ vertices; this implies that $x,y,z\leq \frac{n}{2}$.
	
	Next, consider any $H$-component $X$ of class $1$-$3$ that has $x+1$ vertices, hence $x-1$ internal vertices. Further, consider any $H$-component $Y$ of class $2$-$3$ that has $y+1$ vertices, hence $y-1$ internal vertices. By (P2) the roots $r_X$ and $r_Y$ of $X$ and $Y$ are side edges of two (not necessarily distinct) designated edges $e_i$ and $e_j$, respectively. Assume that $i\leq j$, as the case in which $i>j$ can be discussed analogously. Refer to Figures~\ref{fig:decomposition}(c) and~\ref{fig:decomposition}(d).
	
	Recall that $x\leq \frac{n}{2}$ has been proved already. Suppose, for a contradiction, that $y>\frac{n-x}{2}$. By (P1) we have $e_{j+1}\neq r_X,r_Y$. By construction, $K_{j+1}$ has a number of (internal) vertices which is larger than or equal to the number of (internal) vertices of $Y$. Since $X$, $Y$, and $K_{j+1}$ do not share internal vertices, they have a total of at least $(x-1)+(y-1)+(y-1)>(x-1)+n-x-2=n-3$ internal vertices; these are also internal vertices of $G$. Further, $G$ contains at least two more internal vertices. Namely:
	\begin{itemize} 
		\item if $i<j$, then the apexes of $e_i$ and $e_j$ incident to $r_X$ and $r_Y$, respectively, are internal vertices of $G$; and 
		\item if $i=j$, then there are at least two apexes of $e_i$ which are internal vertices of $G$, given that $e_i$ has at least three side edges, namely $r_X$, $r_Y$, and $e_{i+1}$, and given that any apex of $e_i$ is incident to two side edges of $e_i$. 
	\end{itemize} 
	It follows that $G$ has more than $n-1$ internal vertices, a contradiction. This proves that $y\leq \frac{n-x}{2}$ and concludes the proof of the lemma.
\end{proof}

We now give an algorithm to construct a planar straight-line drawing $\Gamma$ of~$G$. Let $f(n)=n^{\log_2 \phi}$, where $\phi = \frac{1+\sqrt 5}{2}$. The algorithm receives as input a triangle $a_1a_2a_3$, whose hypotenuse $\overline{a_1a_2}$ is such that $||\overline{a_1a_2}||\geq f(n)$, and constructs a planar straight-line drawing $\Gamma$ of $G$ satisfying the following properties: \begin{enumerate}
	\item[(T0)] the length of every edge is at least $1$ and at most $||\overline{a_1a_2}||$;
	\item[(T1)] for $i=1,2$, the vertex $v_i$ lies at $a_i$; and
	\item[(T2)] every internal vertex of $G$ lies inside $a_1a_2a_3$.
\end{enumerate} 

If $n=1$, that is, $G$ coincides with the edge $v_1v_2$, then $\Gamma$ is the straight-line segment $\overline{a_1a_2}$. Then property (T1) is trivially satisfied, property (T2) is vacuous, and property (T0) is satisfied since $||\overline{a_1a_2}||\geq n^{\log_2 \phi} = 1$. 

\begin{figure}[htb]\tabcolsep=4pt
	\centering
	\includegraphics[scale=1.2]{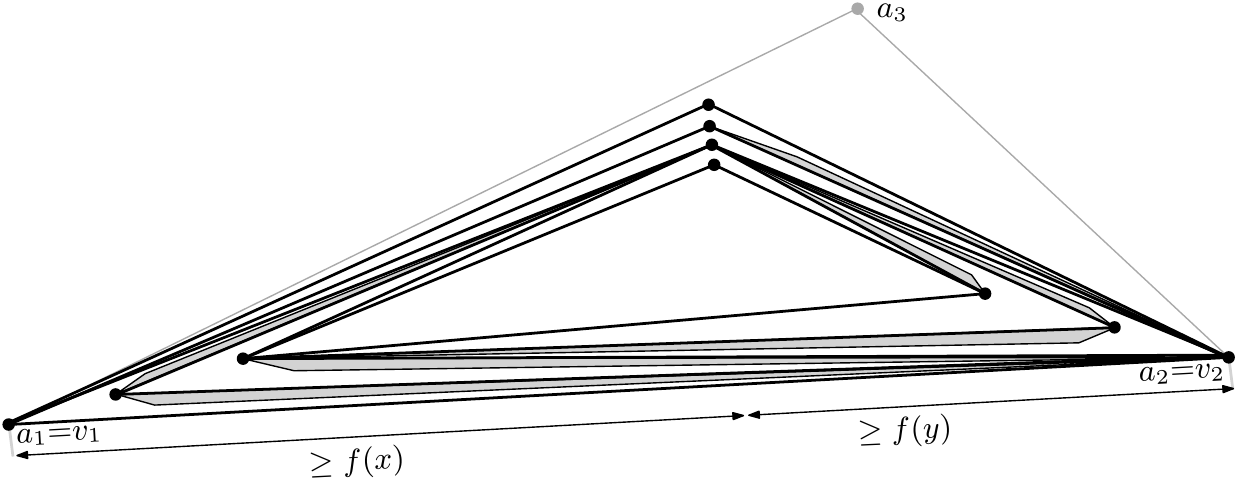}
	\caption{The planar straight-line drawing $\Gamma$ of the graph $G$ in Figure~\ref{fig:2-trees} constructed by the algorithm in the proof of Theorem~\ref{th:2-trees}.}
	\label{fig:2-trees-algorithm}
\end{figure}

Assume next that $n>1$ and refer to Figure~\ref{fig:2-trees-algorithm}. Let $H$ be a subgraph of $G$ satisfying the properties of Lemma~\ref{le:h-component}; in particular (i) $x\leq \frac{n}{2}$ and $y\leq \frac{n-x}{2}$, or (ii) $y\leq \frac{n}{2}$ and $x\leq \frac{n-y}{2}$, where $x$ and $y$ are the maximum number of vertices of an $H$-component of $G$ of class $1$-$3$ and $2$-$3$, respectively, minus one. We construct a planar straight-line drawing $\Gamma_H$ of $H$ by applying the algorithm L2T drawer with input the triangle $a_1a_2a_3$ and the real values $\ell_{1\textrm{-}3}=f(x)$, $\ell_{2\textrm{-}3}=f(y)$, and $\ell_{1\textrm{-}2}=f(z)$. Note that $\ell_{1\textrm{-}2}=f(z)<f(n)\leq ||\overline{a_1a_2}||$, given that $z<n$. Further, as proved by Chan~\cite{DBLP:journals/algorithmica/Chan02}, we have $f(n)=n^{\log_2 \phi}\geq x^{\log_2 \phi}+y^{\log_2 \phi}=f(x)+f(y)$, given that $x\leq \frac{n}{2}$ and $y\leq \frac{n-x}{2}$, or $y\leq \frac{n}{2}$ and $x\leq \frac{n-y}{2}$; this implies that $||\overline{a_1a_2}||\geq \ell_{1\textrm{-}3} + \ell_{2\textrm{-}3}$. 

Let $G_1,\dots,G_k$ be the $H$-components of $G$; for $i=1,\dots,k$, let $u_iv_i$ be the root of $G_i$. Note that $u_iv_i$ is an edge of $H$, hence it is represented by a straight-line segment $\overline{u_iv_i}$ in $\Gamma_H$. For $i=1,\dots,k$, let $w_i$ be a point such that the triangle $\Delta_i=u_iv_iw_i$ lies inside $a_1a_2a_3$, does not intersect $\Gamma_H$ other than at $\overline{u_iv_i}$, and does not intersect any distinct triangle $\Delta_j$, except at common vertices. Since $\Gamma_H$ is planar, choosing $w_i$ sufficiently close to $\overline{u_iv_i}$ suffices to accomplish these objectives. For $i=1,\dots,k$, we recursively draw $G_i$ so that $u_i$ and $v_i$ lie at the same points as in $\Gamma_H$ and so that every internal vertex of $G_i$ lies inside $\Delta_i$. This concludes the construction of a planar straight-line drawing $\Gamma$~of~$G$.

We prove that $\Gamma$ satisfies properties (T0)--(T2). Property~(T1) is satisfied since $\Gamma_H$ satisfies property~(L1); further, property~(T2) is satisfied since $\Gamma_H$ satisfies property~(L2), since the internal vertices of $G_i$ lie inside the triangle $\Delta_i$, and since $\Delta_i$ lies inside $a_1a_2a_3$, by construction. We now deal with property~(T0). The length of every edge of $H$ in $\Gamma$ is at least $\min \{f(x),f(y),f(z)\}$ by property~(L3) of $\Gamma_H$; further, $f(x)=x^{\log_2 \phi}\geq 1$, $f(y)=y^{\log_2 \phi}\geq 1$, and $f(z)=z^{\log_2 \phi}\geq 1$, given that $x,y,z\geq 1$. The length of every edge of $H$ in $\Gamma$ is at most $||\overline{a_1a_2}||$, given that every vertex of $H$ lies inside or on the boundary of $a_1a_2a_3$, by properties~(L1) and~(L2) of $\Gamma_H$, and given that $\overline{a_1a_2}$ is the hypotenuse of $a_1a_2a_3$. The length of every edge of $G$ not in $H$ is at least $1$ and at most $||\overline{a_1a_2}||$ by induction and since every triangle $\Delta_i$ lies inside $a_1a_2a_3$. This completes the induction. 

Applying the described algorithm with a triangle $a_1a_2a_3$ whose hypotenuse has length $||\overline{a_1a_2}||=f(n)$ results in a planar straight-line drawing of $G$ with edge-length ratio at most $f(n)=n^{\log_2 \phi}$. This concludes \mbox{the proof of Theorem~\ref{th:2-trees}.} 

We remark that $f(n)=n^{\log_2 \phi}$ is the smallest possible function when using the decomposition of Lemma~\ref{le:h-component}, as an example in which $x=\frac{n}{2}$ and $y=\frac{n}{4}$ shows.

We also remark that a graph has treewidth at most $2$ if and only if it is a subgraph of a $2$-tree; hence, Lemma~\ref{le:subgraphs} and Theorem~\ref{th:2-trees} imply the following.

\begin{corollary} Every $n$-vertex graph with treewidth at most $2$ has planar edge-length ratio in $O(n^{\log_2 \phi})\subseteq \areaSP$, where $\phi=\frac{1+\sqrt 5}{2}$ is the golden ratio. 
\end{corollary}

The bound on the treewidth in the above result is the best possible, as the proof of Theorem~\ref{th:lower-bound} shows that an $n$-vertex planar graph with treewidth $3$ might have planar edge-length ratio in $\Omega(n)$. Observe that graphs with treewidth $1$, i.e., trees, have  planar edge-length ratio equal to $1$.

\subsection{Bipartite Planar Graphs} \label{se:bipartite}

In this section we deal with bipartite planar graphs, which we prove to have planar edge-length ratio arbitrarily close to $1$.

\begin{theorem} \label{th:bipartite}
	For every $\epsilon>0$, every bipartite planar graph has planar edge-length ratio smaller than $1+\epsilon$.
\end{theorem}

\begin{proof}
	First, it suffices to prove the statement for \emph{maximal} bipartite planar graphs. This follows by Lemma~\ref{le:subgraphs} and by the fact that any nonmaximal bipartite planar graph can be augmented to be maximal by adding edges to it. 
	
	\begin{figure}[htb]\tabcolsep=4pt
		\centering
		\begin{tabular}{c c c c}
			\includegraphics[scale=1.25]{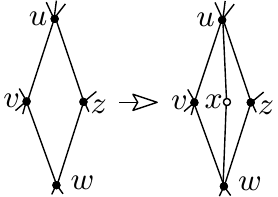} \hspace{1mm} &
			\includegraphics[scale=1.25]{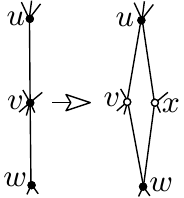} \hspace{1mm} &
			\includegraphics[scale=1.25]{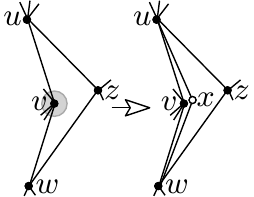} \hspace{1mm} &
			\includegraphics[scale=1.25]{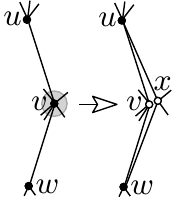} \\
			(a) \hspace{1mm} & (b)\hspace{1mm} & (c)\hspace{1mm} & (d)\\
		\end{tabular}
		\caption{(a) The operation $P_0$. (b) The operation $P_1$. (c) and (d) show how to transform a drawing $\Gamma'$ of $G'$ into a drawing $\Gamma$ of $G$ by applying the operations $P_0$ or $P_1$, respectively. The gray disk~is~$D$.}
		\label{fig:bipartite}
	\end{figure}
	
	Second, Brinkmann et al.~\cite{DBLP:journals/dm/BrinkmannGGMTW05} proved that every $n$-vertex maximal bipartite plane graph $G$ is either a $4$-cycle embedded in the plane, or can be obtained from an $(n-1)$-vertex maximal bipartite plane graph $G'$ by applying either the operation $P_0$ shown in Figure~\ref{fig:bipartite}(a), in which a path $uxw$ is inserted in a face $f$ of $G'$ delimited by a $4$-cycle $uvwz$, or the operation $P_1$ shown in Figure~\ref{fig:bipartite}(b), in which a path $uvw$ of $G'$ is transformed into a $4$-cycle $uvwx$. 
	
	We now prove that, for every $\epsilon>0$, any $n$-vertex maximal bipartite plane graph $G$ admits a planar straight-line drawing $\Gamma$ in which every edge has length larger than $1$ and smaller than $1+\epsilon$. The proof is by induction on $n$. If $n=4$, then $G$ is a $4$-cycle embedded in the plane, and the desired drawing $\Gamma$ of $G$ is any square with side length equal to $1+\delta$, with $0<\delta<\epsilon$.
	
	If $n>4$, then let $G'$ be an $(n-1)$-vertex maximal bipartite plane graph such that $G$ can be obtained from $G'$ by applying either the operation $P_0$ or the operation $P_1$. Fix any $\delta$ such that $0<\delta<\epsilon$; inductively construct a planar straight-line drawing $\Gamma'$ of $G'$ in which every edge has length larger than $1$ and smaller than $1+\delta$. Let $\ell_1=\min_e \{\ell_{\Gamma'}(e)-1\}$, $\ell_2=\min_e \{1+\epsilon -\ell_{\Gamma'}(e)\}$, and $\ell=\min \{\ell_1,\ell_2\}$, where the first two minima are over all the edges $e$ of $G'$. Let $D$ be a disk with radius $\ell$ centered at $v$ in $\Gamma'$. 
	
	
	Both the operations $P_0$ and $P_1$ correspond to the \emph{expansion} of a vertex $v$ into an edge $vx$, followed by the removal of such an edge. Hence, it follows from standard continuity arguments that a planar straight-line drawing $\Gamma$ of $G$ can be obtained from $\Gamma'$ by suitably replacing the vertex $v$ with the edge $vx$, so that the position of $v$ in $\Gamma$ is the same as in $\Gamma'$, and so that $x$ is arbitrarily close to $v$ in $\Gamma$; see, e.g., the proof of F\'ary's theorem~\cite{fary}.
	
	Thus, both if the operation $P_0$ or if the operation $P_1$ transform $G'$ into $G$, we can obtain a planar straight-line drawing $\Gamma$ of $G$ in which every vertex other than $x$ is at the same position as in $\Gamma'$, and in which $x$ is inside the disk $D$; see Figures~\ref{fig:bipartite}(c) and~\ref{fig:bipartite}(d). Note that, for every edge $e$ of $G$ that is not incident to $x$, we have $1<\ell_{\Gamma}(e)<1+\epsilon$, given that $1<\ell_{\Gamma'}(e)<1+\delta$. Further, consider any edge $e=tx$ of $G$ and note that $e'=tv$ is an edge of $G'$. By the triangular inequality we have  $||\overline{tx}||<||\overline{tv}||+||\overline{vx}||<\ell_{\Gamma'}(e')+\ell\leq \ell_{\Gamma'}(e')+(1+\epsilon-\ell_{\Gamma'}(e'))=1+\epsilon$, and $||\overline{tx}||>||\overline{tv}||-||\overline{vx}||>\ell_{\Gamma'}(e')-\ell\geq \ell_{\Gamma'}(e')-(\ell_{\Gamma'}(e')-1)=1$. This concludes the induction and hence the proof of the theorem.
\end{proof}

Note that the bound in Theorem~\ref{th:bipartite} is the best possible, as there exist bipartite planar graphs (for example any complete bipartite graph $K_{2,m}$ with $m\geq 3$) that admit no planar straight-line drawing with edge-length ratio equal to $1$.

\section{Conclusions and Open Problems}\label{se:conclusions}

In this paper we have proved that there exist $n$-vertex planar graphs whose planar edge-length ratio is in $\Omega(n)$; that is, in any planar straight-line drawing of one of such graphs, the ratio between the length of the longest edge and the length of the shortest edge is in $\Omega(n)$. Further, we have proved upper bounds for the planar edge-length ratio of several graph classes, most notably an $\areaSP$ upper bound for the planar edge-length ratio of $2$-trees. Several problems remain open; we mention some of them. 

\begin{itemize}
	\item First, what is the asymptotic behavior of the planar edge-length ratio of $2$-trees? In particular, we wonder whether our geometric construction can lead to a better upper bound if coupled with a decomposition technique better than the one in Lemma~\ref{le:h-component}. 
	
	\item Second, is the planar edge-length ratio of cubic planar graphs sub-linear? The proof of Theorem~\ref{th:lower-bound} shows that this question has a negative answer when extended to all bounded-degree planar graphs. 
	
	\item Third, is the planar edge-length ratio of $k$-outerplanar graphs bounded by some function of $k$? The results from~\cite{DBLP:journals/tcs/LazardLL19} show that this is indeed the case for $k=1$. 
\end{itemize}

Our final remark is about nonplanar drawings. A drawing of a graph is \emph{proper} if no two vertices overlap with each other and no edge overlaps with a nonincident vertex. The \emph{edge-length ratio of a graph} $G$ is the minimum edge-length ratio of any (possibly nonplanar) proper straight-line drawing of $G$. The next theorem shows that the edge-length ratio of a graph is tied, up to multiplicative constants and lower order terms, to the square of its chromatic number.

\begin{theorem} \label{th:non-planar}
	Let $G$ be a graph. If $G$ has chromatic number $k$, then it has edge length ratio at most $c_1\sqrt k + O(1)$, for some constant $c_1$. Further, if $G$ has edge length ratio $h$, then it has chromatic number at most $c_2\cdot h^2 + O(h)$, for some constant $c_2$.
\end{theorem}

\begin{proof}
	First, let $G$ be a graph with chromatic number $k$; refer to Figure~\ref{fig:nonplanar}(a). Let $C_1,\dots,C_k$ be the color classes in a proper $k$-coloring of $G$. Consider a $\lceil \sqrt k \rceil \times \lceil \sqrt k \rceil$ grid $\mathcal H$ and let $p_1,\dots,p_k$ be any $k$ distinct points of $\mathcal H$ . For $i=1,\dots,k$, consider a disk $D_i$ centered at $p_i$ with radius $\epsilon/3$, for some small constant $\epsilon>0$; pick a set $S_i$ of $|C_i|$ distinct points in the interior of $D_i$, in such a way that no three points in $S_1\cup \dots \cup S_k$ are collinear. For $i=1,\dots,k$, injectively map the vertices in $C_i$ to the points in $S_i$. The resulting straight-line drawing $\Gamma$ of $G$ is proper, as no three points in $S_1\cup \dots \cup S_k$ are collinear. Further, the length of any edge of $G$ in $\Gamma$ is larger than $1-2\epsilon/3$; this comes from the triangular inequality and from the fact that the distance between any two centers of distinct disks $D_i$ and $D_j$ is larger than or equal to~$1$. Finally, the length of any edge of $G$ in $\Gamma$ is smaller than $\sqrt{2k}+2\epsilon/3$; this comes again from the triangular inequality and from the fact that the distance between any two centers of distinct disks $D_i$ and $D_j$ is smaller than or equal to~$\sqrt{2k}$, given that the side length of $\mathcal H$ is $\lceil \sqrt k \rceil -1 < \sqrt k$. It follows that the edge-length ratio of $\Gamma$ is smaller than or equal to $c_1\sqrt k + O(1)$, with $c_1=\sqrt 2/(1-2\epsilon/3)$. 	
	
	\begin{figure}[tb]\tabcolsep=4pt
		\centering
		\begin{tabular}{c c}
			\includegraphics[scale=1.1]{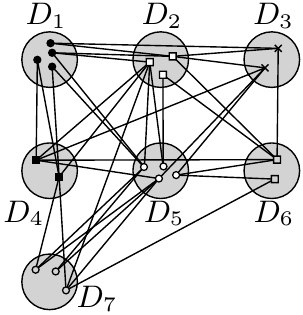} \hspace{6mm} &
			\includegraphics[scale=1.1]{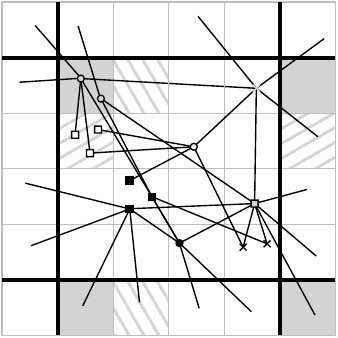} \\
			(a) \hspace{6mm} & (b) \\
		\end{tabular}
		\caption{(a) Illustration for the proof that a graph with chromatic number $k$ has edge length ratio at most $c_1\sqrt k + O(1)$, for some constant $c_1$; in this example $k=7$. (b) Illustration for the proof that a graph with edge-length ratio $h$ has chromatic number at most $c_2\cdot h^2 + O(h)$, for some constant $c_2$. Only the vertices inside a big square are shown, together with their incident edges. Thick and thin lines delimit big and small squares, respectively.}
		\label{fig:nonplanar}
	\end{figure}

	Second, let $G$ be a graph that admits a proper\footnote{In this proof we actually only use the fact that the length of every edge in $\Gamma$ is greater than $0$.} drawing $\Gamma$ with edge-length ratio $h$; refer to Figure~\ref{fig:nonplanar}(b). Assume, w.l.o.g.\ up to a uniform scaling, that the minimum length of any edge in $\Gamma$ is $1$ (and hence the maximum length of any edge in $\Gamma$ is $h$). Let $\ell = h+1+\epsilon$, for some small constant $\epsilon>0$ such that $\sqrt 2 \cdot \ell$ is not an integer. Tessellate the plane by squares with side length $\ell$; these squares are in the following called \emph{big squares}. Partition each big square into $\lceil \sqrt 2 \cdot \ell\rceil$ rows and $\lceil \sqrt 2 \cdot \ell\rceil$ columns of \emph{small squares}. For each big square $Q$, color the small square at the $i$-th row and $j$-th column of $Q$ with color $c_x$, where $x=(i-1)\cdot \lceil \sqrt 2 \cdot \ell\rceil + j$; then assign to each vertex of $G$ the color of the small square it lies into (vertices on the boundary of several small squares are colored with the color of any of the small squares they belong to). The resulting coloring of $G$ uses $\lceil \sqrt 2 \cdot \ell\rceil^2<(\sqrt 2 \cdot  (h+1+\epsilon)+1)^2$ colors; this is at most $c_2\cdot h^2 + O(h)$, with $c_2=2$. We now prove that the constructed coloring of $G$ is proper; that is, any two vertices $u$ and $v$ of $G$ with the same color are not adjacent. First, if $u$ and $v$ are in the same small square $S$, then their distance is smaller than or equal to the length of the diagonal of $S$. The latter is smaller than $1$, given that $S$ has side length $\frac{\ell}{\lceil \sqrt 2 \cdot \ell\rceil}<\frac{1}{\sqrt 2}$ (the inequality uses the fact that $\sqrt 2 \cdot \ell$ is not an integer); hence, $u$ and $v$ are not adjacent. Second, if $u$ and $v$  are not in the same small square, then they are at a distance greater than or equal to the sum of the side lengths of $(\lceil \sqrt 2 \cdot \ell\rceil-1)$ small squares. This value is $\frac{(\lceil \sqrt 2 \cdot \ell\rceil-1) \cdot\ell}{\lceil \sqrt 2 \cdot \ell\rceil}=\ell -\frac{\ell}{\lceil \sqrt 2 \cdot \ell\rceil}>h+\epsilon$. Since the length of any edge of $G$ in $\Gamma$ is at most $h$, it follows that $u$ and $v$ are not adjacent. 
\end{proof}	

A corollary of Theorem~\ref{th:non-planar} and of the four-color theorem is that planar graphs admit proper straight-line drawings with edge-length ratio in $O(1)$. This is in contrast to our $\Omega(n)$ lower bound for the planar edge-length ratio of $n$-vertex planar graphs shown in Theorem~\ref{th:lower-bound}.


\bibliographystyle{splncs03} 
\bibliography{bibliography}

\end{document}